\documentclass{article}
\usepackage{graphicx,subfigure}
\usepackage{amsfonts,amsbsy,amscd,amsgen,amsthm,dsfont,amsmath,gensymb}
\usepackage[numbers,sort&compress]{natbib}
\usepackage[colorlinks=true,allcolors=blue]{hyperref}
\usepackage{bm}
\theoremstyle{definition}
\usepackage{tabularx}
\usepackage{enumitem}
\usepackage[utf8]{inputenc}
\usepackage[T1]{fontenc}
\usepackage{authblk}
\usepackage{placeins}
\usepackage{color}
\usepackage{xspace}
\usepackage{appendix}

\newcommand{\id}{\mathrm d}

\newcommand{\co}{CO\textsubscript{2}\xspace}

\DeclareMathAlphabet\mathbfcal{OMS}{cmsy}{b}{n}

\newtheorem{theorem}{\bf Theorem}[section]

\newtheorem{corollary}{\bf Corollary}[section]
\newtheorem{proposition}{\bf Proposition}[section]

\begin{document}

\title{Investigating climate tipping points under various emission reduction and carbon capture scenarios with a stochastic climate model\thanks{Accepted for publication in Proc. Roy. Soc. A}}
\author[]{Alexander Mendez} 
\author[]{Mohammad Farazmand\thanks{Corresponding author's email address: farazmand@ncsu.edu}}
\affil{Department of Mathematics, North Carolina State University,
	2311 Stinson Drive, Raleigh, NC 27695-8205, USA}
\date{}
\maketitle

\begin{abstract}We study the mitigation of climate tipping point transitions using an energy balance model. 
The evolution of the global mean surface temperature 
is coupled with the \co concentration through the green house effect. We model the \co concentration
with a stochastic delay differential equation (SDDE), accounting for various carbon emission and capture scenarios.
The resulting coupled system of SDDEs exhibits a tipping point phenomena: if \co concentration exceeds a critical
threshold (around $478\,$ppm), the temperature experiences an abrupt increase of about six degrees Celsius.
We show that the \co concentration exhibits a transient growth which may cause a climate tipping point,
even if the concentration decays asymptotically.
We derive a rigorous upper bound for the \co evolution which quantifies its transient and asymptotic growths, and
provides sufficient conditions for evading the climate tipping point. 
Combining this upper bound with
Monte Carlo simulations of the stochastic climate model, we investigate the emission reduction and carbon capture
scenarios that would avert the tipping point.
\end{abstract}

\section{Introduction}
It is highly likely that anthropogenic greenhouse gases are responsible for more than half of the increase in the global mean surface temperature between 1951 to 2010~\cite{AR5}.
Therefore, reducing the atmospheric concentration of greenhouse gases must be a central component of any climate change mitigation strategy.
This reduction can be achieved in two ways.
One involves reducing \co emissions through alternative energy sources, increased fuel efficiency, and reduced consumption~\cite{AR5SPM}.
A second approach is expanding carbon sinks, e.g., by employing carbon capture and storage (CCS) technologies, which capture \co from large emitting sources (such as factories) or directly from the atmosphere and prepares it for long-term storage~\cite{Gibbins, Haszeldine}.

What complicates these matters is the possible existence of climate \emph{tipping points}, i.e., climate regimes where small changes significantly alter the future state of the system~\cite{Drijfhout2015,Lentonetal,schneider2019}. This tipping behavior can involve the sudden disruption of climatological and ecological processes such as melting of ice sheets or large-scale death of rainforests~\cite{Lenton}.
Even if emission reduction and carbon capture strategies lead to long-term \co reduction, its transient growth can trigger a climate tipping point
with adversely irreversible impact. Therefore, mitigation strategies not only have to ensure long-term reduction of greenhouse gasses, but also ensure that their transient response remains below the critical tipping point levels.

The purpose of the present study is to quantify emission reduction and carbon capture scenarios that mitigate climate tipping points.
To this end, we use a Budyko--Sellers-type model of the climate~\cite{Budyko,Sellers}. This energy balance model assumes that the Earth's radiation output is balanced by radiation input, and represents key aspects of glacial-interglacial climate transitions~\cite{Hogg}. It relies on the hypothesis that glacial cycles are triggered by variations in the Earth's orbit which alter the incoming solar radiation to the earth. These orbital cycles are amplified by limited feedback between greenhouse gases and temperature~\cite{McGuffie}. The resulting governing equations are a one-way coupling between 
the global mean surface temperature and \co concentration, where the temperature is affected by \co concentration through the greenhouse effect.

We model the evolution of the \co concentration by a stochastic delay differential equation (SDDE),
which takes into account the carbon emission and capture rates. This model allows us to parameterize the
possible reduction in \co emission rates as well as the capability of carbon sinks to absorb atmospheric \co.
In particular, we show that, even when the emission rates are reduced, the \co concentration exhibits a transient growth which may instigate a climate tipping point. We investigate the response of global mean surface temperature to various emission reduction and carbon capture scenarios, identifying the scenarios which will mitigate the climate tipping point transition. 

\subsection{Related work}
The study of climate tipping points can be divided into three broad categories: modeling, prediction, and mitigation.From the modeling perspective, it is important to accurately parameterize various contributing factors such as green house gasses, water vapor, clouds, and ice sheets~\cite{schneider1974,oerlemans1984,Dessler2013,zelinka2020}. Although here we use a very simple energy balance model, it takes into account the main
culprit, i.e., \co emissions. Furthermore, we choose the model parameters such that the resulting climate sensitivity agrees with the available estimates from more elaborate models.

Prediction is concerned with early warning signs embedded in observational data that may indicate an upcoming climate tipping point~\cite{Scheffer2012}. Several such precursors have been proposed, including critical slowing down~\cite{scheffer2008} and increased variability~\cite{Lenton2012}. Moreover, it has been observed that the stochastic component of the system changes its characteristics close to a tipping point. For instance, Held and Kleinen~\cite{Held2004} and Livina and Lenton~\cite{Livina2007} model the North Atlantic thermohaline circulation and note that the stochastic component of the data changes from white noise to red noise. The same transition is observed by Prettyman et al.~\cite{Prettyman2018} who studied tropical cyclones.
 
From the mitigation standpoint, it is widely accepted that a significant reduction in \co emissions is a necessity~\cite{cai2016,anderson2019,schneider2019}. A complementary solution is to remove \co from the atmosphere using carbon capture technologies~\cite{haszeldine2009}. For instance, studying the Atlantic thermohaline circulation (THC), Bahn et al.~\cite{Bahn2011} find that a drastic and fast \co reduction is required to avoid disrupting the THC. However, it is argued by Ritchie et al.~\cite{Ritchie2019} that crossing a threshold will not necessarily result in a simultaneous extreme change in the climate system. They find that the determining factor is how far the tipping point threshold is exceeded and for how long. 

Previous studies mainly focus on the asymptotic climate state as a result of increasing \co concentration. In contrast, a main focus of our work is the transient climate dynamics in response to various emission reduction and carbon capture scenarios. In particular, we show that, under certain emission reduction and carbon capture scenarios, the \co concentration exhibits a transient growth large enough to trigger a climate tipping point, even though the concentration decays asymptotically. We note that, beyond emission reduction and carbon capture, geoengineering ideas have been proposed~\cite{keith2000}. These methods seek to increase the reflected energy of the Sun by releasing aerosol particles into the stratosphere. These geoengineering ideas are beyond the scope of the present work and are not considered here.

\subsection{Outline of this paper}
This paper is organized as follows. We discuss the climate model and its parameters in sections~\ref{sec:model} and~\ref{sec:params}. In section~\ref{sec:analysis}, we derive a rigorous upper bound for the \co levels under various emission reduction and carbon capture scenarios.
We discuss the temperature variations under each scenario in section~\ref{sec:results} using direct numerical simulations. Section~\ref{sec:conc} contains our concluding remarks.
\section{Stochastic climate model}\label{sec:model}
We use two stochastic differential equations for the climate system, modeling the global mean surface temperature, $T$, and average concentration of CO\textsubscript{2}, $C$, in the atmosphere.
The equation for temperature is a Budyko-Sellers-type model, derived from the balance between incoming and outgoing radiations~\cite{Budyko,Sellers}. The resulting temperature model reads,
\begin{equation}
\label{eqn:BGS}
c_T \frac{\id T}{\id t}= F(T,C) :=  Q_0 (1 - \alpha(T)) + S +  A \ln\left(\frac{C}{C_p}\right) - \epsilon\sigma T^4,
\end{equation}
where the temperature $T$ is in units of Kelvin, $C$ is the concentration of CO\textsubscript{2} in parts per million (ppm), and $c_T$ is the thermal inertia in units of $\mbox{Jm}^{-2}\mbox K^{-1}$. 
The term $Q_0 (1 - \alpha(T))$ represents short-wave radiation from the Sun, where $Q_0$ is the solar input in units of $\mbox{Wm}^{-2}$. The multiplier $\alpha(T)$ denotes temperature dependent albedo that accounts for the light reflecting off the Earth surface. The term $S +  A \ln(C/C_p)$ models the effect of greenhouse gases, where $A$ (in units of $\mbox{Wm}^{-2}$) is the direct forcing of CO\textsubscript{2} and determines the sensitivity of the climate equilibrium~\cite{Dijkstra2015}. The parameter $S$ represents the trapping of outgoing radiations by greenhouse gases~\cite{Hogg}. 
The constant $C_p$ denotes the preindustrial concentration of \co in units of ppm. 
Finally, the last term $ - \epsilon\sigma T^4$  represents long-wave radiation from the Sun, where $\sigma T^4$ represents the outgoing long wave radiation that is modified by the emissivity $\epsilon$. 

Our choice of the temperature dependent albedo function $\alpha(T)$ shown in Fig.~\ref{fig:albedo} is similar to~\cite{Ashwin_2019}, although we use a different set of parameters as reported in Table~\ref{tab:table}.  
In particular, we define
\begin{equation} 
\alpha(T) = \alpha_1(1 - \Sigma(T)) + \alpha_2\Sigma(T),
\label{eq:albedo}
\end{equation}
which transitions smoothly between albedo parameters $\alpha_1$ and $\alpha_2$, where
$\alpha_1$ represents the current global albedo, while $\alpha_2$ represents the global albedo of the Earth. The case $\alpha_1>\alpha_2$ corresponds to the melting of ice on the Earth's surface, whereas the case $\alpha_1<\alpha_2$
represents the formation of ice. Here, we only consider the case $\alpha_1>\alpha_2$ which conforms to the present trends.

The melting of ice depends on the temperature levels. The function $\Sigma(T)$ is chosen to reflect this temperature dependence, so that $\alpha(T)$ 
transitions smoothly between the albedo $\alpha_1$ at temperature $T_1$, to the threshold of $\alpha_2$ at temperature $T_2$. More precisely, we define
\begin{equation} 
\Sigma(T)= \frac{T - T_1}{T_2 - T_1} H(T-T_1)H(T_2 - T) + H(T-T_2),
\label{eq:sigma}
\end{equation}
where
\begin{equation} 
H(T) = \frac{1+ \tanh(T/T_\alpha)}{2},
\label{eqn:heaviside}
\end{equation}
is a smooth approximation of the unit Heaviside function. The parameter $T_\alpha$ controls the transition rate between temperatures $T_1$ and $T_2$.
\begin{figure}
	\centering
	\includegraphics[width=0.65\textwidth]{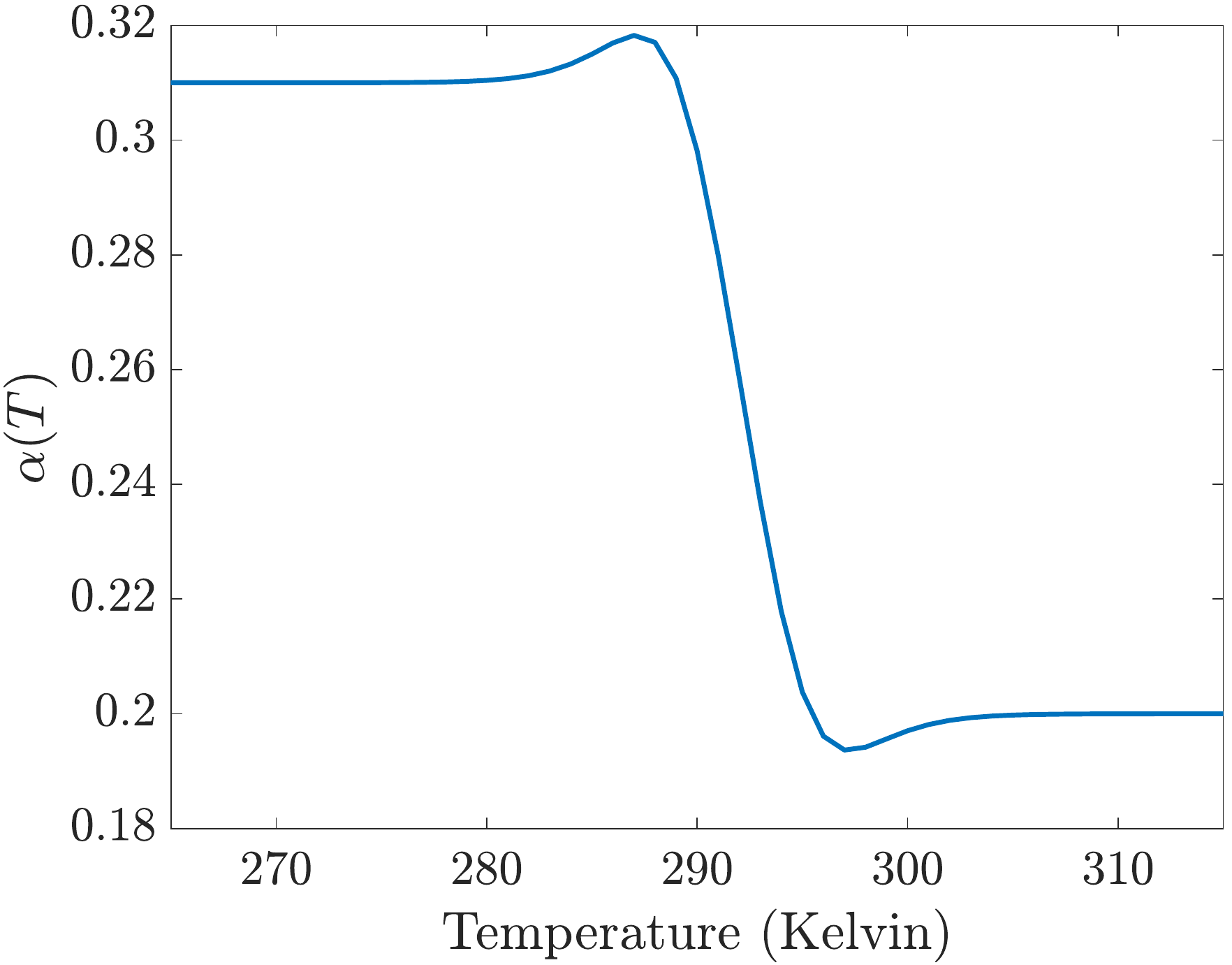}
	\caption{Temperature dependent albedo function $\alpha(T)$ with parameters $T_1 = 289$, $T_2 = 295$, $\alpha_1 = 0.31$, $\alpha_2 = 0.2$, and $T_\alpha = 3$.}
	\label{fig:albedo}
\end{figure}
Following~\cite{Ashwin_2019}, we model the unresolved subgrid processes by a stochastic term $\eta_T\id W_T$, where $\eta_T>0$ is a constant amplitude and $W_T(t)$ denotes the standard Wiener process. The subscript $T$ is used to distinguish the stochastic processes affecting the temperature from those affecting \co, to be described shortly.
Adding this stochastic term, the temperature model reads
\begin{equation}
c_T \id T = F(T,C) \id t + \eta_{T} \id W_T,
\label{eqn:TempSDE}
\end{equation}
where $F(T,C)$ is defined in equation~\eqref{eqn:BGS}. Note that, although the driving stochastic force is a Gaussian white noise, the response $T$ itself is non-Gaussian and non-white owing to the nonlinear nature of $F(T,C)$.

To close the equations, we also need to model the evolution of \co. Dijkstra and Viebahn~\cite{Dijkstra2015} prescribe $C(t)$ as an explicit function of time.
Ashwin and von der Heydt~\cite{Ashwin_2019} model the \co evolution as a random walk confined to an interval $[C_1,C_2]$. To mimic the current trends and to allow for possible 
emission reduction and carbon capture, we model the \co evolution as a stochastic delay differential equation,
\begin{equation} 
\id C= \big[\underbrace{\;\beta(t) C(t)\;}_{\mbox{source}} - \underbrace{GC(t-\tau)}_{\mbox{sink}}\big]\id t + \eta_C \id W_C,
\label{eqn:CO2SDE}
\end{equation}
where $\beta(t)$ is a time-dependent emission rate, $G$ is a time-independent carbon capture rate, and $\tau$ is a time delay.
The uncertainties in \co concentration are modeled with the stochastic term $\eta_C\id W_C$, where $\eta_C$ is a constant noise intensity
and $W_C(t)$ is a standard Wiener process.

The source term in equation~\eqref{eqn:CO2SDE} models the \co emission into the atmosphere. 
We allow for a time-dependent rate $\beta(t)$ to model reduction or enhancement of the \co emissions.
The sink term, on the other hand, models the \co captured from the atmosphere either through natural phenomena (e.g., photosynthesis~\cite{Sedjo2012})
or through artificial technologies (e.g., chemical looping~\cite{Song2018}).
For the \co sinks, we allow for a time delay $\tau$ to model the time between carbon capture and its effect being felt in the atmospheric concentration. This type of delay is typical in control theory, where there is often a lag between a modifying action and its actualization~\cite{dorf2011,Farazmand_2020}. 
If this delay is non-existent or negligible, one can set $\tau=0$. Here, we set $\tau=1$ year, representing the time it takes for CO2 data to be updated and the carbon capture strategies adjusted correspondingly. Nonetheless, we have varied the time delay in the interval $1-20$ years (not shown here) and only observed variations of approximately $10$ ppm in the \co evolution, which do not significantly alter the results.

Equations~\eqref{eqn:TempSDE} and~\eqref{eqn:CO2SDE} form a closed set of equations modeling the global mean temperature $T$ and \co concentration $C$.
Our main goal here is to investigate whether a combination of emission reduction and carbon capture may avert a possible upcoming climate tipping point. 
To this end, we consider various scenarios in terms of emission reduction and carbon capture as discussed in section~\ref{sec:params}.

\section{Tipping points and model parameters}\label{sec:params}
The parameter values used in the model are listed in Table~\ref{tab:table} and their choice is discussed in Appendix~\ref{app:params}.
For a prescribed \co concentration, the temperature model~\eqref{eqn:TempSDE} exhibits three distinct regimes as shown in Fig.~\ref{fig:scurve}. 
Regime 1: If $C<C_1=378$ ppm, the temperature has a single stable equilibrium satisfying $F(T, C)=0$.
\begin{table}
	\centering
	\caption{Model parameters and their physical dimensions. See Appendix~\ref{app:params} for more detail on the choice of parameter values.}
\begin{tabular}{|p{1.4cm}|p{2.3cm}p{1.4cm}|p{1.4cm}|p{1.8cm}p{1.2cm}|}
	\hline
	Parameter & Value  & Units & Parameter & Value  & Units \\
	\hline
	\rule{0pt}{10pt}$c_T$ & $5.0 \cdot 10^8$  & $J m^{-2} K^{-1}$ & $Q_0$ & $342$  & $W m^{-2}$ \\
	$\epsilon$ & $1$ & $-$ & $\sigma$ & $5.67 \cdot 10^{-8}$ &$s^{-1/2}$ \\
	$A$ & $20.5$ & $Wm^{-2}$ & $G_0$ & $2.37 \cdot 10^{-10}$ & $s^{-1}$ \\
	$C_0$ & $410$ & $ppm$& $\alpha_1$ & $0.31$ & $-$ \\
	$\alpha_2$ & $0.2$ & $-$ & $T_{1}$ & $289$ & $K$\\
	$T_{2}$ & $295$ & $K$ & $T_\alpha$ & $3$ & $K$\\
	$a $& $4.30\cdot 10^{-10}$ & $s^{-1}$ &$\eta_T$ & $5\cdot10^{-6}$ & $s^{-1/2}$\\
	$\tau$ & $60\cdot60\cdot24\cdot365$ & $s$ & $S$ & $150$ &$Wm^{-2}$\\
	$T_0$ & $288$ & $K$ & $C_{p}$ & $280$ &$ppm$\\
	$\eta_C$ & $6.5\cdot10^{-7}$ & $s^{-1/2}$&  &  &\\
	\hline
\end{tabular}
\label{tab:table}
\end{table}
\FloatBarrier

Regime 2: For \co concentrations $C \in (C_1,C_2)$ a bifurcation takes place, whereby two additional equilibria are born. 
One of the new equilibria is unstable (dashed line in Fig.~\ref{fig:scurve}) 
whereas the other one is stable. As a result, in this regime the temperature model is bistable.
Regime 3: If $ C>C_2=478.6$, the model switches back to a single stable equilibrium. However, in this regime, the global mean temperature is 
significantly higher than the equilibrium temperature in regime 1.
\begin{figure}[!htb]
	\centering
	\includegraphics[width=0.65\textwidth]{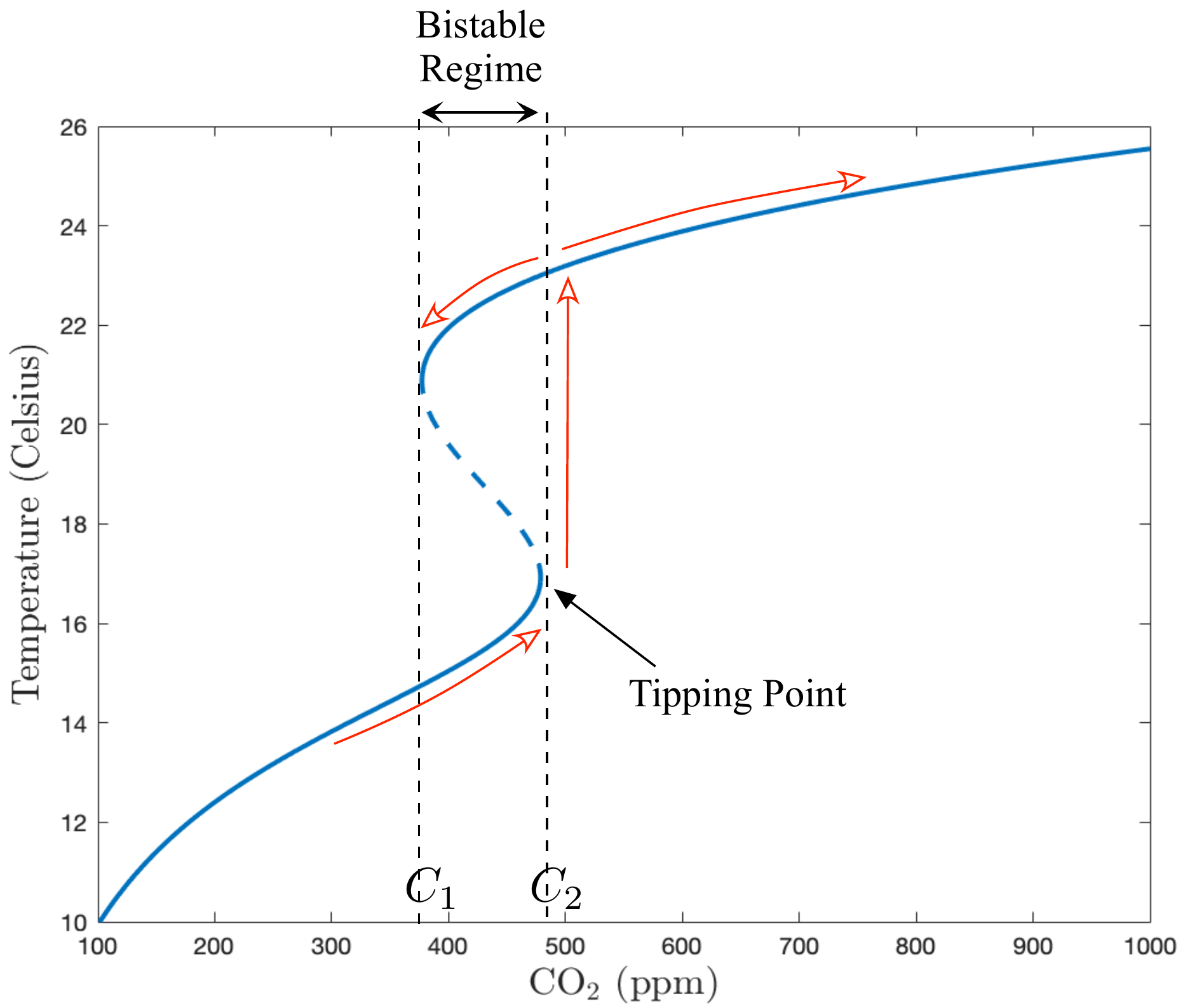}
	\caption{Bifurcation diagram of the temperature $T$ versus the \co concentration $C$. The curve marks the equilibrium states of
		equation~\eqref{eqn:BGS}, obtained by setting $F(T,C)=0$, with parameters in Table~\ref{tab:table}. Solid parts of the curve denote stable equilibria and the dashed part denotes unstable equilibria.}
	\label{fig:scurve}
\end{figure}

In our model, the \co concentration is not constant, instead it evolve according to the SDDE~\eqref{eqn:CO2SDE}. We choose the 
initial conditions $T(0)=15\degree$C and $C(0) = 410\,$ppm which reflect the current global mean temperature and \co concentration, respectively~\cite{NASAtemp, Hansen2010,noaa}. 
The model parameters, listed in Table~\ref{tab:table}, are chosen such that the current climate lies in bistable regime 2 near the lower branch of equilibria.
If the current trends continue, i.e., the \co concentration keeps increasing, the temperature $T$ also continues to increase gradually. Eventually, one reaches the tipping point 
$C=C_2$ where a slight increase in \co concentration leads to a dramatic increase in the temperature by about six degrees in Celsius. Our goal is to determine the emission reduction
and carbon capture scenarios that would avert this catastrophic climate tipping point.

To this end, we allow for a time-dependent emission rate $\beta(t)$ as shown in figure~\ref{fig:beta}. It contains three adjustable periods. First is a period of inaction, up to time $t=t_1$, where the emission rate remains constant at its current level $a$. It is followed by a reduction period, $t_1<t<t_2$ where the emission rate decreases linearly toward its terminal value $b=a/n$. This reduction continues for a period of $\Delta t= t_2-t_1$ years until it plateaus at time $t=t_2$.
The period $t>t_2$ constitutes the terminal stage where the emission rate remains at the constant level $b$.
\begin{figure}
	\centering
	\includegraphics[width=0.65\textwidth]{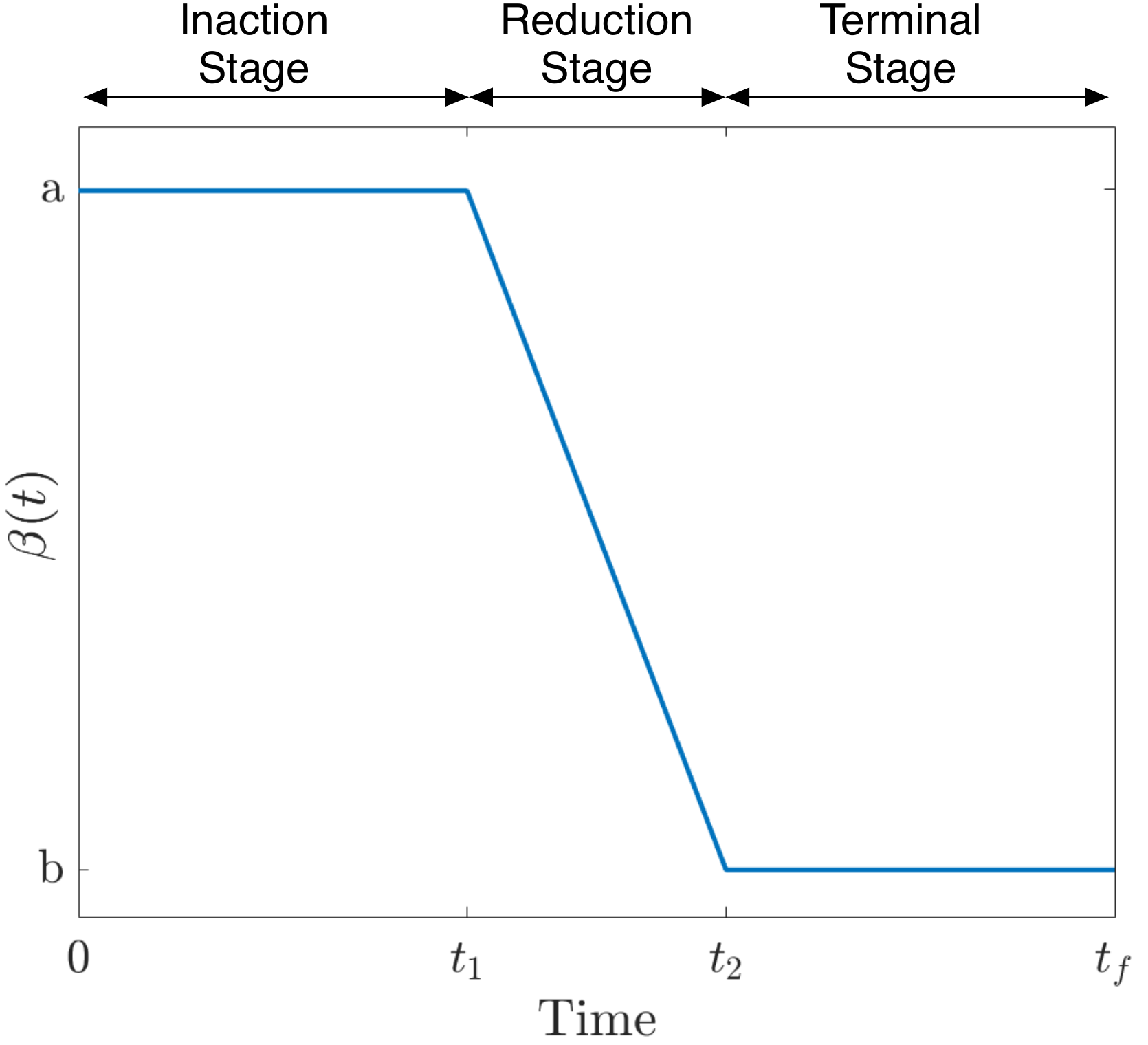}
	\caption{The carbon emission rate $\beta(t)$. We assume the rate is initially equal to 
		$a$ and remains so until time $t_1$. It is then followed with a linear reduction over the time window $\Delta t = t_2-t_1$.
		Eventually the emission rate reaches the plateau $b =a/n$ for $t>t_2$.}
	\label{fig:beta}
\end{figure}

Therefore, the emission rate has the general form
\begin{equation} \label{eq:beta}
\beta (t) =
\begin{cases}
a, & \quad 0 \leq t < t_1, \\
a + \frac{b-a}{t_2-t_1}(t - t_1) , & \quad t_1 \leq t \leq t_2, \\
b, & \quad t_2 < t \leq t_f,
\end{cases}
\end{equation}
where $b=a/n$ for an integer $n\geq 2$. The current emission rate $a$ is estimated from the data in Ref.~\cite{essd}. They estimate that there are $11.8 \pm 0.9$ gigatons of carbon emissions each year from industrial processes and land usage. Converting this to concentration yields $5.556 \pm 0.4237$ ppm added to the atmosphere each year. The current emission rate $a$ is then estimated from 
$aC(0) = 5.556/(365\cdot24\cdot60\cdot60)$ ppm/s, where $C(0) \simeq 410$ ppm is the current \co concentration.
This yields $a = 4.2971\cdot 10^{-10}\; \mbox s^{-1}$. In the following sections, we study the evolution of the global mean temperature $T$ under various choices of the remaining free 
parameters, $t_1$, $t_2$ and $n$.

To conclude the choice of parameters, we need to specify the parameters $G$ and $\tau$ in the \co model~\eqref{eqn:CO2SDE}.
Recall that the parameter $G$ denotes the capacity of carbon sinks to remove \co from the atmosphere. Assuming that carbon capture takes place only due to 
natural phenomena, the value of $G$ can be estimated from the
data in Ref.~\cite{essd}. They estimate that approximately $6.5$ gigatonnes of carbon are taken from the atmosphere each year from natural processes such as photosynthesis or absorption by the oceans. Converting this to concentration yields $3.0603$ ppm removed from the atmosphere each year. The parameter $G$, then can be estimated by 
$GC(0) = 3.0603/(365\cdot24\cdot60\cdot60)$ ppm/s, which yields $G = 2.3669 \cdot 10^{-10}\;\mbox s^{-1}$. 

Finally, we allow for a delay $\tau$ in the \co sinks, which can be set to zero if no such delay exists. We have not been able
to determine this parameter from the available literature. The results in Section~\ref{sec:results} are reported for the delay time $\tau$ equal to one year. However, we examined delay times up to three years and did not observe a significant effect on the results.
\section{Transient and asymptotic growth of \co}\label{sec:analysis}
The system of equations~\eqref{eqn:TempSDE} and~\eqref{eqn:CO2SDE} is a one-way coupling between the temperature $T$ and the \co concentration $C$, whereby
the \co concentration evolves independently of the temperature. As discussed in section~\ref{sec:params}, if the \co concentration increases beyond $C_2=478.6\,$ppm, the climate system undergoes a 
tipping point transition leading to an abrupt increase of the global mean surface temperature (see Fig.~\ref{fig:scurve}). 

The main objective of this paper is to determine emission reduction and carbon capture scenarios that ensure the mitigation of this tipping point. 
Extreme events, such as tipping point transitions, have been the subject of much research, with an emphasis on their causal mechanisms~\cite{ansmann2013,Farazmande1701533,farazmand2019a}, 
probabilistic quantification~\cite{weinan2006,weinan2010,Dematteis2018,mohamad16}, and
data-driven prediction~\cite{bialonski2015,Farazmand2017,scheffer2010,scheffer2008}. Only recently, control strategies for mitigating extreme events have been proposed~\cite{bialonski2015,Farazmand_2020,farazmand2019b}. 
In particular, Farazmand~\cite{Farazmand_2020} proposes a time-delay feedback control
for mitigating noise-induced transitions in multistable systems. This control strategy relies on the stationary equilibrium density of the system. Given the time-dependent emission rate $\beta(t)$
and the linearity of the \co model~\eqref{eqn:CO2SDE}, it does not possess such an equilibrium density. Consequently, the framework of Ref.~\cite{Farazmand_2020} is not applicable here.

Therefore, we take a different approach here and derive a quantitative upper bound for the \co concentration which determines
whether the climate tipping point transitions can be averted.
To this end, we consider the \co model~\eqref{eqn:CO2SDE} without the stochastic term,
\begin{equation} 
\frac{\id C}{\id t} := \beta(t) C(t) - GC(t-\tau),\quad C(s) = C_0,\quad \forall s\in[-\tau,0].
\label{eqn:CO2DE}
\end{equation}
Note that~\eqref{eqn:CO2DE} is a delay differential equation which requires the initial condition $C(s)$ to be specified as a function over the 
interval $[-\tau,0]$. Since we consider a relatively short delay $\tau$ of one year, the \co concentration does not change significantly over this
time interval. As a result, we assume the constant initial condition $C(s) = C_0$ for all $s\in [-\tau,0]$, where $C_0=410$ ppm is the current level of \co concentration.
The following theorem provides an upper bound for the solutions $C(t)$. Adding the stochastic term $\eta_C\id W_C$ only leads to small fluctuations around this upper bound, without fundamentally altering the results.
\begin{theorem}\label{thm:main}
Let $0 < T \leq \infty$, and $C_0 \geq 0$. Assume that $\beta: [0, T) \to \mathbb{R}^+$ is locally integrable and that
$C:[-\tau,T)\to \mathbb R^+$ is a measurable, locally integrable function solving the delay differential equation~\eqref{eqn:CO2DE}. Then
\begin{equation}
C(t) \leq C_0\exp{\left(\int_{0}^{t} \big(\beta(s) + \gamma(s) \big)\id s \right)},
\label{eq:ub}
\end{equation}  
where 
\begin{equation}
\gamma(t) = -G\exp{\left(-\int_{t-\tau}^{t} \beta(s) \id s\right)}.
\label{eq:gamma}
\end{equation}
\end{theorem}
\begin{proof}
See Appendix~\ref{app:ub}.
\end{proof}

Recall that the climate tipping point occurs if the \co levels exceed $C_2=478.6$ ppm. Therefore, ensuring that the upper 
bound~\eqref{eq:ub} is uniformly below $C_2$ is a sufficient condition for avoiding the tipping point.
This upper bound encapsulates the competition between \co emission rate $\beta(t)$ and the carbon capture rate $G$. 
The negative-valued function $\gamma(t)$, appearing in the exponent of the upper bound, is proportional to the carbon capture rate $G$. 
But it also depends on the emission rate $\beta(t)$ due to the delay $\tau$. This leads to a non-trivial dependence of the upper bound on 
the carbon emission and capture rates. In section~\ref{sec:results}, we investigate the shape of this upper bound for various parameter values.

The upper bound~\eqref{eq:ub} also provides a sufficient condition for the asymptotic decay of \co concentration. 
\begin{corollary}\label{cor:asym}
Assume the conditions of Theorem~\ref{thm:main}. If the carbon capture rate satisfies $G>be^{b\tau}$, then we have
$\lim_{t\to \infty} C(t) =0$
\end{corollary}
\begin{proof}
Note that for $t>t_2$, we have $\beta(t) = b$ (see equation~\eqref{eq:beta}). Therefore, for all $t>t_2+\tau$, upper bound~\eqref{eq:ub} implies
$$C(t) \leq C_0 \exp\left[\int_0^{t_2+\tau}(\beta(s) + \gamma(s)) \id s\right]\exp\left[\left( b-Ge^{-b\tau}\right)(t-t_2-\tau)\right].$$
As a result, if $G>be^{b\tau}$, the \co concentration $C(t)$ tends to zero as $t\to \infty$.	
\end{proof}

We emphasize that, to avert the climate tipping point, it is not sufficient for the \co concentration to decay asymptotically. 
As we show in Section~\ref{sec:results}, even a transient growth of \co that exceeds the critical threshold $C_2$ will lead to a tipping point transition.

\section{Results and discussion}\label{sec:results}
In this section, we present the numerical results obtained from the model~\eqref{eqn:TempSDE}-\eqref{eqn:CO2SDE}.
First, we focus on the \co model~\eqref{eqn:CO2SDE} and investigate the behavior of its upper bound~\eqref{eq:ub}. Recall that to avoid the climate tipping point it is sufficient for this upper bound 
to remain below the critical level $C_2$.

Figure~\ref{fig:bound} shows the upper bound~\eqref{eq:ub} 
with the delay time $\tau = 1$ year, emission reduction time span $\Delta t = 50$ years, 
and the carbon capture rate $G =G_0= 2.37 \cdot 10^{-10}\;\mbox s^{-1}$.
The initial \co level is set at $C(0) = C_0 = 410\,$ppm, which is the estimated \co concentration in the year 2019~\cite{noaa}. Figure~\ref{fig:a} represents the optimistic scenario where emission reductions begin immediately ($t_1=0$)
and continue to decrease linearly for $\Delta t=50$ years. The figure shows three terminal emission rates $b=a/n$ with $n=2,5,10$. 
In every case, we observe a transient growth of \co levels. But none of them reach the tipping point $C_2$ and they decay asymptotically. 
\begin{figure}
	\centering
	\subfigure[]{\label{fig:a}\includegraphics[width=0.49\textwidth]{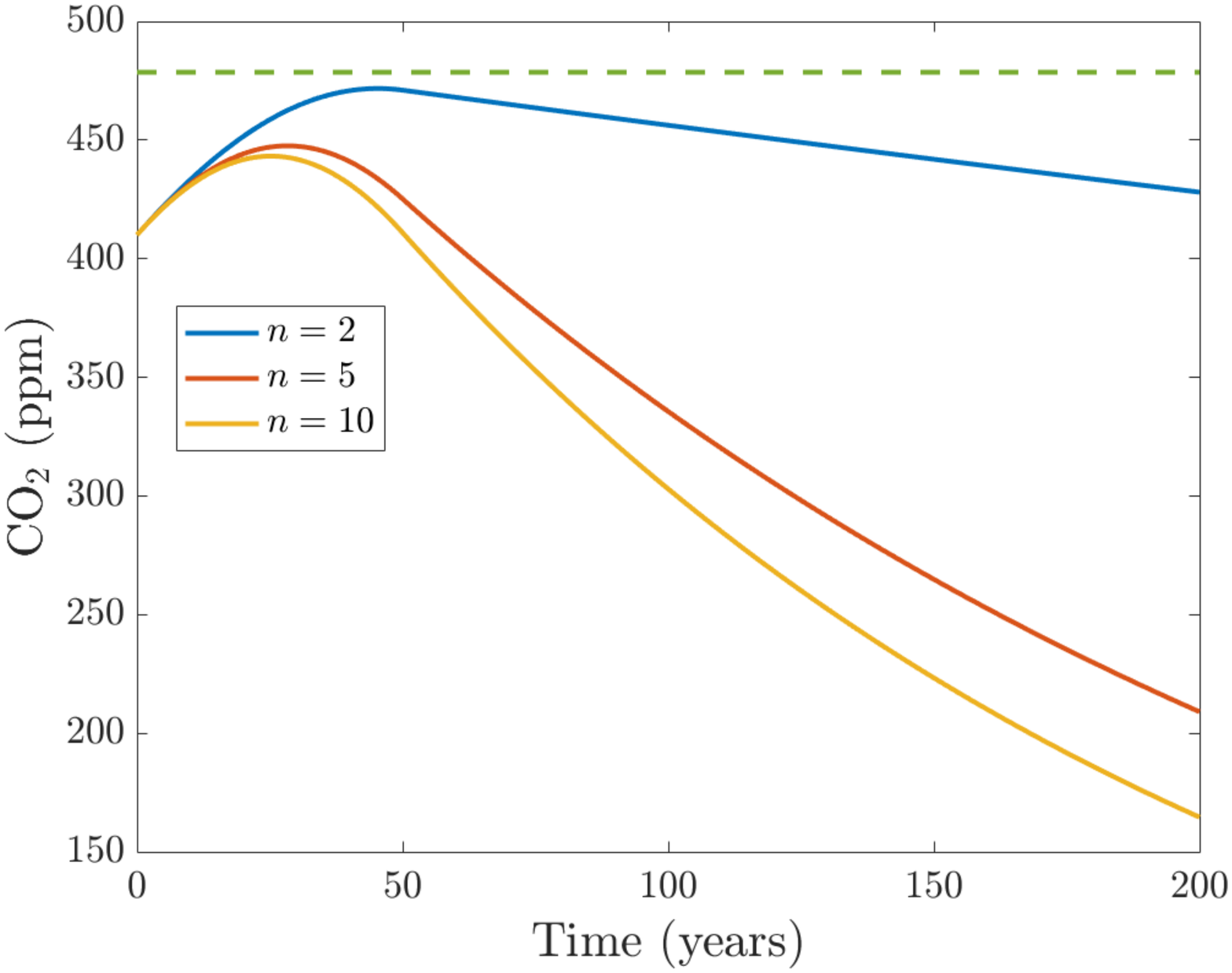}}
	\subfigure[]{\label{fig:b}\includegraphics[width=0.49\textwidth]{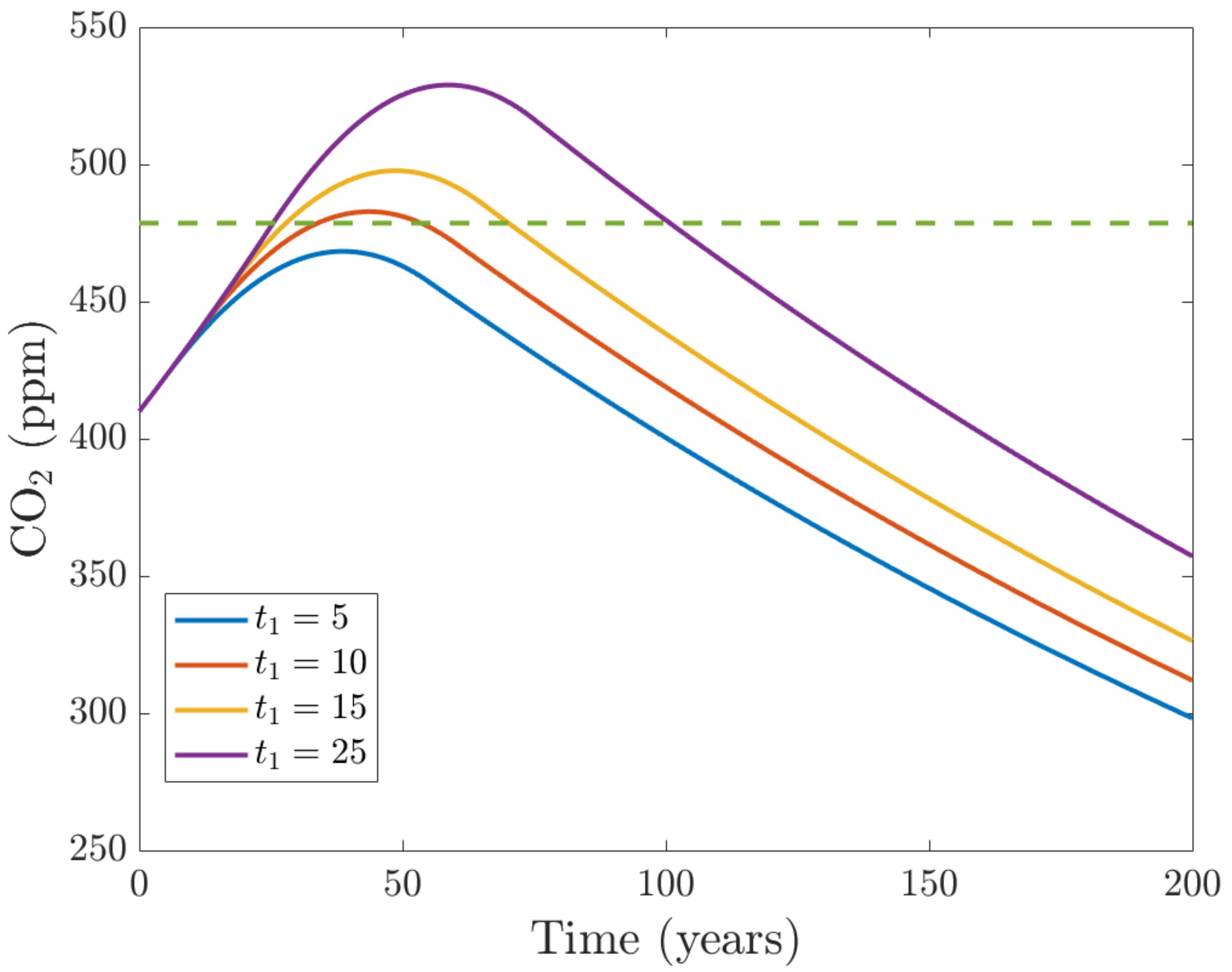}}
	\caption{The \co upper bound~\eqref{eq:ub} for parameters $\tau = 1$ year, $\Delta t = 50$ years, $G = 2.3669 \cdot 10^{-10}\;\mbox s^{-1}$. 
		The horizontal dashed line marks the tipping point $C=C_2 = 478.6$. Time $t=0$ corresponds to the year $2019$.
		(a) $b=a/n$ and $t_1=0$. 
		(b) $b=a/3$ and $t_1>0$.
		}
	\label{fig:bound}
\end{figure}

On the other hand, figure~\ref{fig:b} shows the case where the transient growth surpasses the tipping point $C=C_2$. 
In this figure, we fix the terminal \co emission rate at $b=a/3$ with the time window for reaching this rate being $\Delta t = 50$ years. 
We vary the number of years $t_1$ before the linear reduction in emission rate begins. If reductions begin in $t_1=5$ years, the tipping point can still be avoided. However, if $t_1\geq 10$ years, the tipping point may be reached and consequently the global mean surface temperature may 
sharply increase by about six degrees. These observations accentuate the need for immediate action on carbon emission reduction.

We emphasize that in every case shown in figure~\ref{fig:bound} the \co levels decay asymptotically. 
However, the more germane factor is whether the transient growth of \co  would exceed the climate tipping point
$C=C_2$.
\begin{figure}
	\centering
\includegraphics[width=\textwidth]{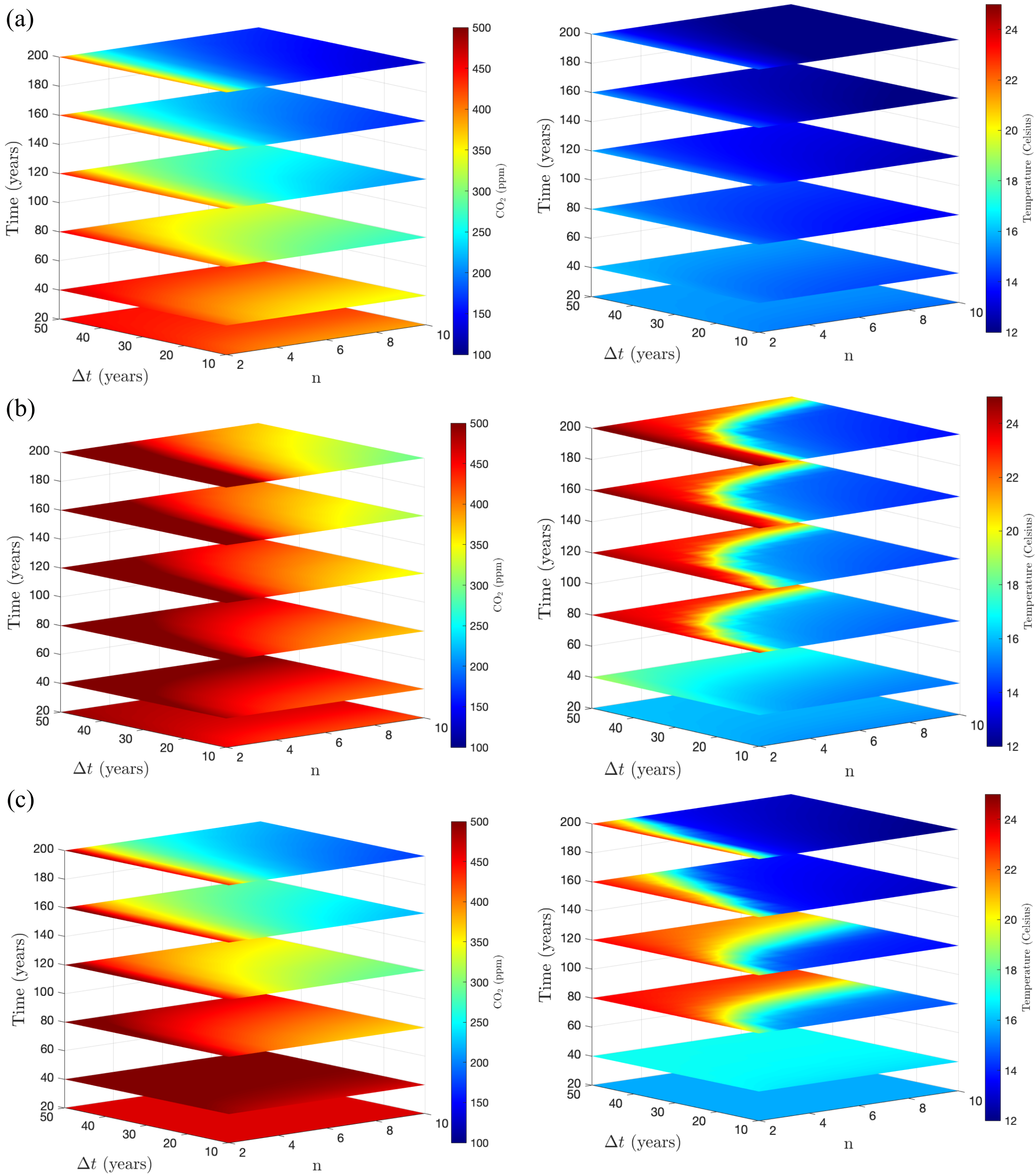}
	\caption{Ensemble averages of the \co concentration $C$ in the first column and the temperature $T$ in the second column, given the terminal emission parameter $n$ and reduction time $\Delta t$. The time $t=0$ corresponds to the year $2019$.
	(a) $t_1=0$ and $G=G_0$. (b) $t_1=0$ and $G=G_0/2$. (c) $t_1=25\,$ years and $G=G_0$.
	Note the transient decrease in \co concentration after 20 years for appropriate choices of parameters $n$ and $\Delta t$.}
	\label{fig:simulations}
\end{figure}

Next we consider the full model, i.e., the couple system of equations~\eqref{eqn:TempSDE}-\eqref{eqn:CO2SDE}. 
Recall that this is a one-way coupling, where the temperature $T$ is affected by the \co concentration through the green house effect.
In contrast, the \co concentration evolves independently of the temperature.

This system of stochastic delay differential equations is integrated numerically using the 
predictor-corrector scheme developed by Cao et al.~\cite{cao2015}.
The numerical integration is carried out after nondimensionalizing the equations by defining 
$\hat C = C/C_{p}$, $\hat T = T/T_0$, and $\hat t = t/t_0$, where $C_{p}=280\,$ppm denotes the preindustrial \co level, $T_0 = 288\,$K denotes the emissivity threshold given by Ashwin and von der Heydt~\cite{Ashwin_2019}, and $t_0=10^7\,$s is an arbitrary time scale (approximately one-third of a year). 
The time step of the numerical integrator in the non-dimensional time $\hat t $ is $0.0086$, which is equivalent to one day in dimensional time.

To investigate the behavior of the climate system with regards to various emission reduction scenarios, 
we simulate the model over a 200 year time span for a range of parameters $n$ and $\Delta t$. 
Recall that $n$ determines the terminal emission rate $b=a/n$ and $\Delta t$ determines the time it takes to reach this plateau (see figure~\ref{fig:beta}). 
For each set of parameters $(n,\Delta t)$, we simulate $10^3$ realizations of the stochastic climate model and compute the ensemble average
of the global mean surface temperature $T$ and the \co concentration. 
In all simulations, the initial conditions are set to $T(0) = 15^\circ\,$Celsius
and $C(0) =410\,$ppm, which reflect the estimated values in the year 2019~\cite{NASAtemp, Hansen2010}.

Figure~\ref{fig:simulations} shows the ensemble means of \co concentration $C$ and temperature $T$ as a function of time and the 
parameters $(n,\Delta t)$. This figure has three panels as described below:
\begin{enumerate}
	\item Figure~\ref{fig:simulations}(a): $t_1=0$ and $G=G_0= 2.37 \cdot 10^{-10}\,\mbox{s}^{-1}$.
	\item Figure~\ref{fig:simulations}(b): $t_1=0$ and $G= G_0/2= 1.18 \cdot 10^{-10}\,\mbox{s}^{-1}$.
	\item Figure~\ref{fig:simulations}(c): $t_1=25\,$years and $G=G_0= 2.37 \cdot 10^{-10}\,\mbox{s}^{-1}$.
\end{enumerate}
Figures~\ref{fig:simulations}(a)-(b) correspond to the scenario where emission reductions begin immediately ($t_1=0$). 
In panel (a), the carbon capture rate is equal to the empirical value $G_0=2.37 \cdot 10^{-10}\,\mbox{s}^{-1}$. In this case, the 
\co concentration increases transiently, but does not reach its critical value $C_2$. As a result, no tipping point phenomena 
is observed. In fact, the temperature does not exceed $16$ degrees Celsius and decays towards $12$ degrees asymptotically.

Panel (b) is identical to (a) except that we use a lower carbon capture rate $G= G_0/2=1.18 \cdot 10^{-10}\,\mbox{s}^{-1}$. The 
carbon capture rate may decrease over time due to various factors, but most prominently due to deforestation~\cite{Maldi2004,baccini2012}. 
In this case, the \co concentration still undergoes a transient growth before eventually decreasing. Because of the lower carbon capture rate, however,
the transient growth surpasses the critical threshold $C_2$, and leads to a drastic increase in the global mean surface temperature.
This tipping point does not occur for all carbon emission scenarios. If the transition period $\Delta t$ is short enough and the terminal 
emission rate $b=a/n$ is small enough, the tipping point can still be mitigated. 

In fact, as shown in figure~\ref{fig:standardeviation}(a), there is a nonlinear boundary in the parameter space $(n,\Delta t)$ which separates the 
tipping point transitions from its successful mitigation. For $n=2$, the tipping point occurs regardless of the transition time $\Delta t$. 
Therefore, the terminal carbon emission rate has to reduce to at least one-third of its current value to mitigate the climate tipping point. Larger $n$ allows for longer transition period $\Delta t$.
In figure~\ref{fig:standardeviation}(b), we show the ensemble average of the temperature $T$ 
for three parameter values $(n,\Delta t)$. The shaded areas mark one standard deviation from the mean. The point on the boundary $(n=4,\Delta t=29)$ has a large standard deviation because the stochastic process $\eta_c\id W_c$ can kick the trajectory towards the tipping regime or away from it. The parameters $n=3$ and $\Delta t=29$ lead to the tipping behavior with high probability, 
where the temperature rises to about $24^\circ$ Celsius. In contrast, $n=6$ successfully averts the climate tipping point, 
after a slight transient increase in the global mean surface temperature.

\begin{figure}
	\centering
		\includegraphics[width=\textwidth]{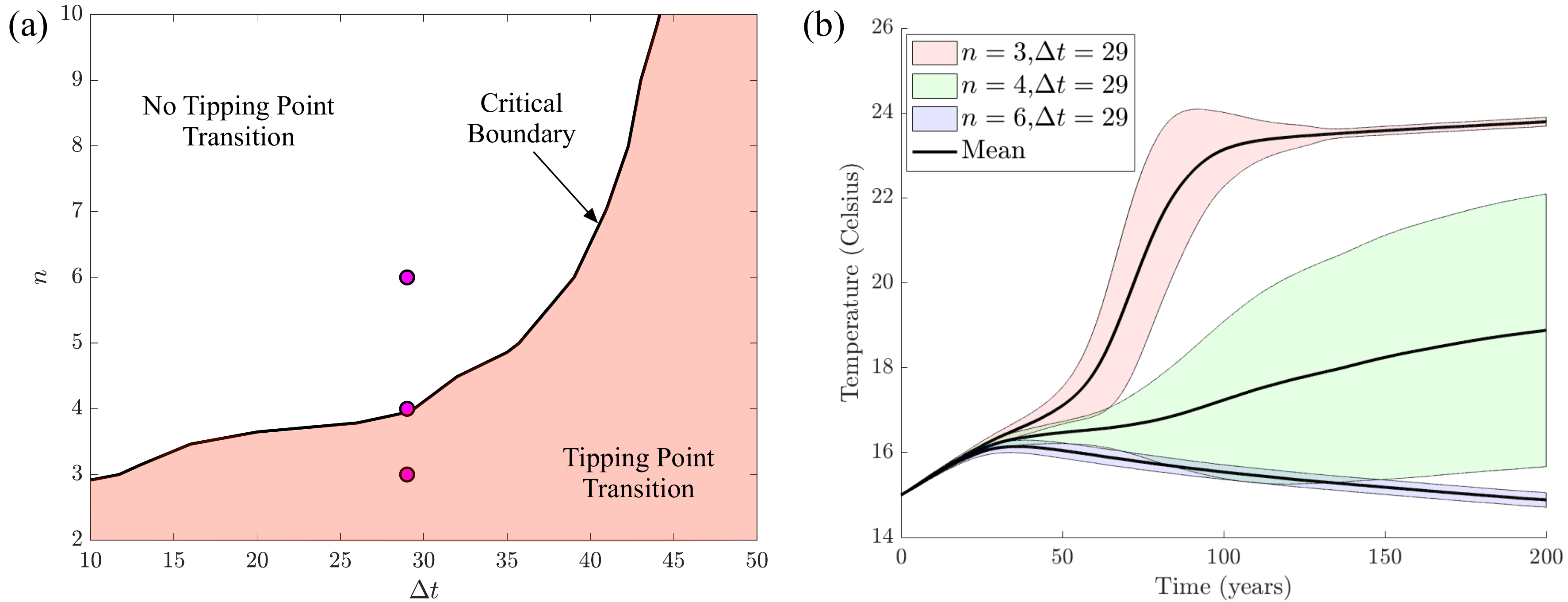}
	\caption{(a) The critical boundary in the parameter space $(n,\Delta t)$ demarcating the climate tipping point regime.
		(b) Ensemble average of the temperature for three parameter values marked by circles in panel (a). The shaded regions mark one standard deviation around the mean.}
	\label{fig:standardeviation}
\end{figure}

Finally, figure~\ref{fig:simulations}(c) shows the case where emission reduction is delayed by 25 years ($t_1=25$), and the 
carbon capture rate stays at its current level ($G=G_0$). In this case, \co concentration again exhibits a transient growth past the critical value
$C_2$, and as a result, a drastic increase in global temperature ensues. Although the \co concentration reduces asymptotically to as low as $150\,$ppm, 
the transient tipping point will have dire consequences, such as sea level rise and droughts.
This observation reaffirms the need for immediate action towards significant emission reduction. 
Otherwise, the only alternative would be to significantly increase the \co sinks $G$, through artificial carbon capture technologies.

\section{Conclusions}\label{sec:conc}
We considered a stochastic climate model as a one-way coupling between \co concentration $C$ and the global mean surface temperature $T$. 
The \co concentration is governed by a stochastic delay differential equation, allowing for the modeling of various emission reduction and carbon capture scenarios. The temperature $T$ satisfies a stochastic differential equation derived from energy balance (Budyko--Sellers model). The temperature is coupled with the \co equation to model the effect of greenhouse gasses.

The model has a tipping point behavior: if the \co concentration exceeds the critical value of $C=478.6$ ppm, the global mean surface temperature will increase abruptly by about six degrees Celsius. The \co model exhibits transient growth which allows for reaching this tipping point in finite time, even when \co decreases asymptotically. We derived a tight upper bound for the \co concentration, which provides sufficient conditions for mitigating the climate tipping point (see Theorem~\ref{thm:main}). This upper bound depends on several parameters such as the \co emission rate $\beta(t)$, the carbon capture rate $G$, and the delay time $\tau$. However, since the upper bound is explicitly known, it can be analyzed numerically with low computational cost in order to determine the emission reduction and carbon capture scenarios which would mitigate the climate tipping point.

We examined various emission reduction scenarios by combining our analytic results with Monte Carlo simulations of the climate model. In particular, we find that the climate tipping point in our model can be averted if \co emission reductions begin immediately 
and the emission rate decreases to one-third of its current level within 50 years. However, if these reductions are delayed by as much as ten years, the transient growth of the \co concentration will exceed the tipping point value, leading to a drastic and abrupt increase in the global mean surface temperature. In the latter case, increasing the carbon capture rate could still avert the tipping point.

It remains to be seen whether our conclusions carry over to more sophisticated climate models, such as box models or general circulation models. We will investigate such models in future work. Although the temperature model would be more complex, our analysis of the \co concentration can be used to derive sufficient conditions that guarantee tipping point evasion.
Furthermore, deriving an invariant probability distribution for the \co model is of great interest as it would quantify the probability 
of transitions which in turn leads to necessary and sufficient conditions for mitigating the climate tipping point.
Finally, the \co model itself can be improved, e.g., by using data-driven discovery methods to obtain simple models that agree with observational data~\cite{Brunton2016b}.\vskip6pt

\enlargethispage{20pt}

\section*{Data availability}
The data and code for reproducing the results of this manuscript are publicly available at \href{https://github.com/mfarazmand/ClimateMitigation.git}{https://github.com/mfarazmand/ClimateMitigation}

\section*{Acknowledgments}
We are grateful to Prof. Tam\'as T\'el for recommending reference~\cite{Ashwin_2019}, and Prof. Pedram Hassanzadeh for fruitful conversations.
This work was partially supported by the National Science Foundation grant DMS-1745654.


\begin{appendices}
\section{Proof of Theorem~\ref{thm:main}}\label{app:ub}
We begin by writing equation~\eqref{eqn:CO2DE} as the integral equation,
\begin{equation}
C \left (t \right ) = C_0+\int_{0}^{t} \beta \left(s\right)C\left(s\right) \id s - \int_{0}^{t} GC\left(s-\tau\right) \id s.
\label{eqn:integral}
\end{equation}
We derive a tight upper bound for $C\left(t\right)$ using a modified version of a Gr\"onwall-type inequality for delay differential equations given in~\cite{Horvath}. We first prove a result with $\beta\equiv 0$, which we later use to address the general case where the emission rate $\beta$ is non-zero.
\begin{proposition}\label{prop:ub}
	Let $t_0 \in \mathbb{R}$, $t_0 < T \leq \infty$, $C_0 \geq 0$, and $a:[t_0 , T) \to \mathbb{R}^+$ be locally integrable.
	Let $C:[t_0, T) \to \mathbb{R}^+$ be Borel measurable and locally bounded such that
	\begin{equation}
	C\left(t\right) \leq C_0 - \int_{t_0}^{t} a\left(u\right) C\left(u-\tau\right) \id u,
	\label{eq:int_C_01}
	\end{equation}
	and $\gamma:[t_0 - \tau, T) \to \mathbb{R}$ be any locally integrable function that satisfies the inequality,
	\begin{equation}
		-a(t) \exp{\left(-\int_{t-\tau}^{t}\gamma\left(s\right) \id s\right)} \leq \gamma\left(t\right).
		\label{ineq:gamma}
	\end{equation}
	Then
	$$
	C\left(t\right) \leq K\exp{\left(\int_{t_0}^{t} \gamma\left(s\right) \id s\right)},
	$$
	where
	$$
	K :=\max \left\{ C_0\exp \left( \int_{t_0 -\tau}^{t_0} \gamma\left(u\right) \id u \right), \sup_{t_0 - \tau \leq s \leq t_0} \exp \left(\int_{s}^{t_0} \gamma\left(u\right) \id u \right) \right \}.
	$$
\end{proposition}

\begin{proof}
	Note that in the context of our climate model, we have $a(t) = G$.  Gy\"ori and Horv\'ath~\cite{Horvath} proved a similar result with 
	a plus sign in front of the integral in Eq.~\eqref{eq:int_C_01}. We generalize their result to the case 
	with a minus sign in front of the integral as required by our climate model.
	We first define
	\begin{align*}
	y(t) = C\left(t\right)\exp{\left(-\int_{t_0-\tau}^{t} \gamma\left(s\right) \id s\right)}.
	\end{align*}
	Inequality~\eqref{eq:int_C_01} implies
	\begin{align*}
	y(t) &\leq C_0\exp{\left(-\int_{t_0-\tau}^{t}\gamma\left(s\right) \id s\right)} - \exp{\left(-\int_{t_0-\tau}^{t} \gamma\left(s\right) \id s\right)} \\
	&\times \int_{t_0}^{t} a\left(u\right) y\left(u-\tau\right)\exp{\left(-\int_{u-\tau}^{u} \gamma\left(s\right) \id s\right)}\exp{\left(\int_{t_0 - \tau}^{u} \gamma\left(s\right) \id s\right)}\id u.
	\end{align*}
	The inequality~\eqref{ineq:gamma} for $\gamma$ then yields
	\begin{align*}
	y(t) \leq & C_0 \exp{\left(-\int_{t_0 -\tau}^{t} \gamma\left(s\right) \id s\right)} +\nonumber\\
	 & \exp{\left(-\int_{t_0 -\tau}^{t} \gamma\left(s\right) \id s\right)}\int_{t_0}^{t} \gamma\left(u\right) \exp{\left(\int_{t_0 -\tau}^{u} \gamma\left(s\right) \id s\right)}y\left(u -\tau\right) \id u.
	\end{align*}
	Defining, 
	$$
	L := \max \left \{C_0 , \sup_{t_0 - \tau \leq s \leq t_0} C(s)\exp{\left(-\int_{t_0 - \tau}^{s} \gamma\left(u\right) \id u\right)} \right \},
	$$
	we have $y(t) \leq L$ for all $t\in [t_0-\tau,t_0]$. We now prove that, in fact, $y\left(t\right)\leq L$ for all $t\in [t_0-\tau,T)$.

	Let $L_1 > L$. By the definition of $L$, we have
	\begin{align*}
	y(t) \leq L < L_1,\quad  t_0 -\tau \leq t \leq t_0
	\end{align*}
	Since $y$ is continuous on $[t_0,T]$, there exists $q>0$ such that $t_0 + q <T$ and
	\begin{align*}
	y(t) < L_1 ,\quad  t_0 \leq t \leq t_0 + q
	\end{align*}
	Assume there exists $t_1 \in (t_0 + q,T)$ such that $y(t_1)=L_1$. Since $y$ is continuous on $[t_0,T)$, we can assume that $y(t) < L_1$ for all $t \in [t_0,t_1)$.
	
	It then follows from $L_1 > L \geq C_0$ that
	\begin{align*}
	y(t_1) &< C_0\exp{\left(-\int_{t_0 - \tau}^{t_1} \gamma(s) \id s \right)} \\
	&+ \exp{\left(-\int_{t_0-\tau}^{t_1} \gamma(s) \id s \right)}\int_{t_0}^{t_1}\gamma(u)\exp{\left(\int_{t_0-\tau}^{u} \gamma(s) \id s \right)}L_1 \id u \\
	&= C_0\exp{\left(-\int_{t_0 -\tau}^{t_1} \gamma(s) \id s \right)} + L_1(1-\exp{\left(-\int_{t_0}^{t_1} \gamma(s) \id s \right)}) \\
	&= L1 + \exp{\left(-\int_{t_0-\tau}^{t_1}\gamma(s) \id s \right)}\left(C_0 - L_1\exp{\left(\int_{t_0-\tau}^{t_0}\gamma(s) \id s \right)} \right) \\
	&< L_1
	\end{align*}
	This contradicts $y(t_1) = L_1$. Therefore, we must have
	$y(t) < L_1$ for all $t\in[t_0-\tau,T)$
	Since $L_1>L$ is arbitrarily close to $L$ and $y$ is continuous, we have $y(t) \leq L$ for all $t\in[t_0-\tau,T)$.
	
	Finally, $y(t)\leq L$ implies
	\begin{align*}
	C(t) &= y(t)\exp{\left(\int_{t_0 - \tau}^{t} \gamma\left(s\right) \id s\right)} \\
	&\leq L \exp{\left(\int_{t_0 - \tau}^{t} \gamma\left(s\right) \id s\right)} \\
	&= K\exp{\left(\int_{t_0}^{t} \gamma\left(s\right) \id s\right)}
	\end{align*}
\end{proof}

Now we return to equation~\eqref{eqn:integral} and derive an upper bound for $C(t)$ for the general case with $\beta(t)\geq 0$.
\begin{theorem}\label{thm:ub_gen}
	Let $t_0 \in \mathbb{R}$, $t_0 < T \leq \infty$, and $C_0 \geq 0$, $\beta: [t_0, T) \to \mathbb{R}^+$ be locally integrable.
	Let $C:[t_0 - \tau, T ) \to \mathbb{R}^+$ be Borel measurable and locally bounded such that
	\begin{equation}
	C(t) \leq C_0 + \int_{t_0}^{t} \beta\left(u\right) C\left(u\right) \id u - \int_{t_0}^{t} GC\left(u-\tau\right) \id u, 
	\label{eq:C_int_02}
	\end{equation}
	and $\gamma:[t_0 - \tau, T) \to \mathbb{R}$ be any locally integrable function satisfying
	\begin{equation}
	-G \exp{\left(-\int_{t-\tau}^{t}\left(\beta\left(s\right) + \gamma\left(s\right)\right)\right)} \leq \gamma\left(t\right).
	\label{eq:gamma_ineq}
	\end{equation}
	Then
	$$
	C(t) \leq K \exp{\left(\int_{t_0}^{t} \left(\beta\left(s\right) + \gamma\left(s\right)\right)\id s\right)},
	$$
	where 
	$$K := \max \left\{C_0\exp{\left(\int_{t_0 - \tau}^{t_0} \gamma(u) \id u\right)}, \sup_{t_0 - \tau \leq s \leq t_0} C\left(s\right) \exp{\left(\int_{s}^{t_0}\gamma(u) \id u\right)}\right\}.$$
\end{theorem}
\begin{proof}
	This result was first proved in~\cite{Horvath} with a plus sign before the last integral in~\eqref{eq:C_int_02}. 
	We generalize their result to the case 
	with a minus sign in front of the integral as required by our climate model. We also point out that our \co model satisfies~\eqref{eq:C_int_02} with equality.
	
	Applying variation of constants to~\eqref{eq:C_int_02}, we obtain
	\begin{align*} C(t) \leq C_0 \exp{\left(\int_{t_0}^{t} \beta\left(v\right) \id v\right)} - \int_{t_0}^{t} G C\left(u-\tau\right)\exp{\left(\int_{u}^{t} \beta\left(v\right) \id v\right)}\id u.
	\end{align*}
	Defining,
	\begin{align*}
	y(t) = C(t) \exp{\left(\int_{t_0}^{t}-\beta\left(v\right) \id v\right)},
	\end{align*}
	we have
	\begin{align*}
	y(t) \leq C_0 - \int_{t_0}^{t} G \exp{\left(-\int_{u-\tau}^{u} \beta(v) \id v\right)}y\left(u-\tau\right) \id u.
	\end{align*}
	The function,
	\begin{align*}
	t \to G\exp{\left(-\int_{t-\tau}^{t} \beta\left(v\right) \id v\right)},
	\end{align*}
	is locally measurable. Thus, applying Proposition~\ref{prop:ub}, we obtain
	\begin{align*}
	y(t) \leq K \exp{\left(\int_{t_0}^{t} \gamma(s) \id s\right)},
	\end{align*}
	where
	\begin{align*}
	K := \max \left\{C_0\exp{\left(\int_{t_0 - \tau}^{t_0} \gamma(u) \id u\right)}, \sup_{t_0 - \tau \leq s \leq t_0} C\left(s\right) \exp{\left(\int_{s}^{t_0}\gamma(u) \id u\right)}\right\}.
	\end{align*}
	The conclusion then follows immediately.
\end{proof}

Although this theorem provides an upper bound for equation~\eqref{eqn:CO2DE}, the issue of finding a suitable choice of $\gamma(t)$ remains. 
To obtain a tight upper bound, $\gamma(t)$ must be strictly negative. To this end, we choose
\begin{align*}
\gamma(t) = -G\exp{\left(-\int_{t-\tau}^{t} \beta\left(s\right) \id s\right)},
\end{align*}
which is negative definite and satisfies the inequality~\eqref{eq:gamma_ineq}.
Recall that we assumed the constant initial condition $C(s)=C_0$ for all $s \in [-\tau,0]$ (cf. equation~\eqref{eqn:CO2DE}). This observation, together with the fact that $\gamma$ is negative, yields $K = C_0$.
Therefore the upper bound in Theorem~\ref{thm:ub_gen} simplifies to 
\begin{align*}
C(t) \leq C_0\exp{\left(\int_{0}^{t} \left(\beta\left(s\right) + \gamma\left(s\right) \right)\id s \right)},
\end{align*}
which is the desired result of Theorem~\ref{thm:main}.

\section{Model parameters}\label{app:params}
Our model parameters mostly agree with those chosen by Dijkstra and Viebahn \cite{Dijkstra2015}. However, there are a few differences that we justify in this section. These changes are mostly made so that our model parameters more closely approximate the current climate state.

The main differences appear in the albedo function~\eqref{eq:albedo}.
Dijkstra and Viebahn \cite{Dijkstra2015} use the initial albedo value $\alpha_1 = 0.7$ which does not agree with the estimates of the current albedo of the Earth.
We instead use $\alpha_1 = 0.31$ which agrees with the value inferred from empirical satellite observations~\cite{Vonder1971, Stephens2012, Goode2001}. The temperature dependent albedo~\eqref{eq:albedo} requires albedo values to be paired with a corresponding temperature threshold. A reasonable choice is to pair the current albedo of Earth with the temperature $T_1 = 289^{\circ}\,$K, which roughly corresponds to the current global mean surface temperature.

The terminal albedo value $\alpha_2 = 0.2$, with the corresponding threshold temperature $T_2=295\,$K, is chosen to reflect the state of the Earth after significant amounts of warming. We point out that the terminal albedo $\alpha_2=0.2$ is most likely unrealistic since it will require a significant reduction in cloud coverage in addition to sea level rise. However, the model does not allow for much larger values of $\alpha_2$
without losing the climate tipping point. As seen in figure~\ref{fig:scurveloop}, for any $\alpha_2 > 0.22$, the fold in the bifurcation curve ceases to exist and, as a result, the irreversible climate tipping point vanishes.
\begin{figure}
	\centering
	\includegraphics[width=0.5\textwidth]{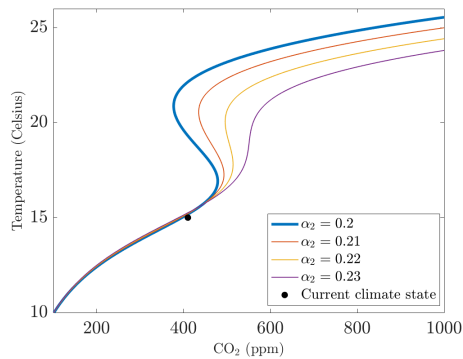}
	\caption{Birfucation curves corresponding to various values of the terminal albedo $\alpha_2$.}
	\label{fig:scurveloop}
\end{figure}
 
Dijkstra and Viebahn~\cite{Dijkstra2015} use the smoothness parameter $T_\alpha = 0.273$ which leads to a rapid change between the initial albedo $\alpha_1$
and the terminal albedo $\alpha_2$. Following Ashwin and von der Heydt~\cite{Ashwin_2019}, we use $T_\alpha=3$ which corresponds to a more gradually decreasing albedo as shown in figure~\ref{fig:albedo}.

We use the \co forcing parameter $A=20.5\, W m^{-2}$. The direct effect of \co forcing yields the lower value $A \simeq 5\,Wm^{-2}$. However, this value neglects other effects such as ice albedo and water vapor~\cite{Hogg,Ashwin_2019}, and results in a climate sensitivity that is much lower than the range given in AR5~\cite{AR5}. We note that Ashwin and von der Heydt~\cite{Ashwin_2019} use $A = 5.35Wm^{-2}$, but they allow for a temperature-dependent emissivity function $\epsilon(T)$ to account for the water vapor feedback. With \co forcing parameter $A=20.5\, W m^{-2}$, the transient climate response (TCR) of our model is about $2.6^\circ$C and its equilibrium climate sensitivity (ECS) is about $6^\circ$C. Note that both the TCR and ECS of our model are only slightly larger than the current likely range given in AR5~\cite{AR5}, which are $1-2.5^\circ$C and $1.5-4.5^\circ$C, respectively.
\end{appendices}


\end{document}